\documentclass[12pt,oneside]{amsart}

\usepackage[utf8]{inputenc}
\usepackage{microtype}
\usepackage{amsfonts}
\usepackage{amsmath}
\usepackage{lmodern}
\usepackage{amsthm}
\usepackage{breqn}
\usepackage{amssymb}
\usepackage{enumerate}

\usepackage{color}

\usepackage{url}

\newtheorem{theorem}{Theorem}[section]
\newtheorem{lemma}[theorem]{Lemma}

\newtheorem{corollary}[theorem]{Corollary}

\theoremstyle{definition}

\newtheorem{example}[theorem]{Example}

\newcommand{\R}{\mathbb{R}}
\newcommand{\Z}{\mathbb{Z}}
\newcommand{\T}{\mathbb{T}}
\newcommand{\p}{\partial}
\newcommand{\x}{{\bf x}}

\newcommand{\ds}{\displaystyle}
\newcommand{\Hy}{\mathbb{H}}

\DeclareMathOperator{\inn}{in}
\DeclareMathOperator{\out}{out}
\DeclareMathOperator{\id}{id}
\DeclareMathOperator{\tr}{tr}
\DeclareMathOperator{\divv}{div}

\DeclareMathOperator{\Vol}{Vol}
\DeclareMathOperator{\vol}{vol}

\usepackage{color}

\usepackage{setspace}
\usepackage[margin=3.5cm]{geometry}

\newcommand\bbZ{{\mathbb{Z}}}
\renewcommand\S{\Sigma}

\author{Michael Eichmair}

\author{Gregory J. Galloway}

\author{Abraão Mendes}

\title{Initial data rigidity results}

\address{Faculty of Mathematics, University of Vienna, Vienna, Austria}
\email{michael.eichmair@univie.ac.at} 

\address{Department of Mathematics, University of Miami, Coral Gables, FL, USA}
\email{galloway@math.miami.edu}

\address{Instituto de Matemática, Universidade Federal de Alagoas, Maceió, AL, Brazil}
\email{abraao.mendes@im.ufal.br}

\begin{document}

\onehalfspacing

\begin{abstract}
We prove several rigidity results related to the spacetime positive mass theorem. A key step is to show that certain marginally outer trapped surfaces are weakly outermost. As a special case, our results include a rigidity result for Riemannian manifolds with a lower bound on their scalar curvature.

\end{abstract}

\maketitle

\section{Introduction} \label{sec:introduction}

In this paper we establish several rigidity results for initial data sets that are motivated by the spacetime positive mass theorem. 

An initial data set $(M, g, K)$ consists of a connected Riemannian manifold $(M, g)$ and a symmetric $(0, 2)$-tensor field $K$. In addition, we assume that $M$ is oriented throughout this paper.

Let $(M, g, K)$ be an initial data set.

The \textsl{local energy density} $\mu$ and the \textsl{local current density} $J$ of $(M, g, K)$ are given by 
\begin {align*}
\mu = \frac{1}{2} \, \left( R - |K |^2 + (\text{tr} \, K)^2 \right) \quad \text{ and } \quad 
J = \text{div} \left(K - (\text{tr} \, K) \, g\right).
\end {align*}
Here, $R$ is the scalar curvature of $(M, g)$. The initial data set is said to satisfy the \textsl{dominant energy condition} (DEC for short)  if 
\[
\mu \geq |\, J \, |.
\]

Let $\Sigma \subset M$ be a two-sided hypersurface with unit normal $\nu$ and 
\[
H = \text{div}_\Sigma \, \nu
\]
be the associated mean curvature. The \textsl{future outgoing null expansion scalar} $\theta^+$ and \textsl{past outgoing null expansion scalar} $\theta^-$ of $\Sigma$ are the quantities
\[
\theta^+ = H + \text{tr}_\Sigma (K) \qquad \text{ and } \qquad \theta^- = H - \text{tr}_\Sigma (K).
\]
The hypersurface $\Sigma$ is \textsl{outer trapped} if $\theta^+ < 0$, \textsl{weakly outer trapped} if $\theta^+ \leq 0$, and \textsl{marginally outer trapped} if $\theta^+ = 0$. 
In the latter case, we refer to $\S$ as a \textsl{marginally outer trapped surface} (MOTS for short). Unless stated otherwise, we require that MOTS are closed, i.e.~compact and without boundary.

We also consider the quantities
\[
\chi^+ = A + K|_\Sigma  \qquad \text{ and } \qquad \chi^- =  A - K|_\Sigma
\]
where $A$ is the second fundamental form of $\Sigma$. Our sign convention is such that $H = \text{tr}_\Sigma \, A$ and thus 
$\theta^{\pm} = {\rm tr}_{\Sigma} \, \chi^{\pm}$. 

Initial data sets  arise naturally in general relativity. Let $M$ be a spacelike hypersurface in a spacetime, i.e.~a time-oriented Lorentzian manifold, $(\bar M, \bar g)$.  Let $g$ be the Riemannian metric induced on $M$ and $K$ be the second fundamental form with respect to the future-pointing unit normal $u$ of $M$ in $\bar M$. Then $(M, g, K)$ is an initial data set. In this setting, $\chi^{+}$ and $\chi^{-}$ are the \textsl{null second fundamental forms} with respect to the \textsl{null normal fields} 
\[
\ell^{+} = \nu + u|_\Sigma \qquad \text{ and } \qquad \ell^{-} = \nu - u|_\Sigma 
\]
of $\S$ viewed as a surface in $\bar M$. Note that $\theta^{\pm} = {\rm div}_\Sigma \,  \ell^{\pm}$.

An initial data set $(M, g, K)$ is said to be \textsl{time-symmetric} or \textsl{Riemannian} if $K = 0$. In this case, the DEC is the requirement that the scalar curvature of $(M, g)$ be non-negative. Moreover, $\S$ is a MOTS if and only if it is a minimal surface in $(M, g)$. Quite generally, MOTS share many properties with minimal surfaces, which they generalize; cf.~e.g.~the survey article \cite{AEM}.

The following version of the spacetime positive mass theorem has been obtained by L.-H.~Huang, D.~A.~Lee, R.~Schoen, and the first-named author in \cite{EichmairHuangLeeSchoen}. 

\begin{theorem} [\cite{EichmairHuangLeeSchoen}]
\label{thm.STposmass}
Let $(M, g, K)$ be an $n$-dimensional asymptotically flat initial data set with ADM energy-momentum vector $(E, P)$. Assume that $3 \leq n \leq 7$. If the dominant energy condition $\mu \geq |J|$ is satisfied, then $E \ge |P|$.
\end{theorem}

We refer to \cite{EichmairHuangLeeSchoen} for the definition of the energy-momentum vector. The case of equality $E = |P|$ has recently been characterized by L.-H.~Huang and D.~Lee \cite{Huang-Lee:2020}.

In \cite{Lohkamp2016}, J.~Lohkamp has presented a different proof of Theorem \ref{thm.STposmass} for all $n \geq 3$. His method is by reduction to and proof of the following result: Let $(M, g, K)$ be an initial data set that is isometric to Euclidean space, with $K = 0$, outside some bounded open set $U \subset M$. Then one cannot have $\mu > |J|$ on $U$; see \cite[Theorem~2]{Lohkamp2016}. In particular, if $(M, g, K)$ satisfies the DEC, there must be a point in $U$ where $\mu = |J|$. The goal of our first result is to show that a much stronger conclusion holds when $3 \le n \le 7$. 

Under the assumption of Lohkamp's result stated above, one obtains by obvious inclusion and identification a compact initial data set $(\tilde M, \tilde g, \tilde K)$ with boundary $\partial \tilde M = \Sigma_1 \cup \Sigma_2$ where $\Sigma_1$ and $\Sigma_2$ are flat $(n-1)$-tori in $(\tilde M, \tilde g)$ that both are totally geodesic in the spacetime sense, i.e.~$\chi^{\pm} = 0$ with respect to either choice of unit normal.
In particular, both are MOTS.  With this compactification (also used by Lohkamp) in mind, we state our first main rigidity result.

\begin{theorem} \label{thm.global.foliation}
Let $(M,g,K)$ be an $n$-dimensional,  $3 \leq n \leq 7$, compact-with-boundary initial data set.  Suppose that $(M,g,K)$ satisfies the DEC, $\mu \geq |J|$. Suppose also that the boundary can be expressed as a disjoint union $\p M=\Sigma_0\cup S$ of non-empty unions of components such that the following conditions hold:
\begin{enumerate}
\item $\theta^+\le0$ on $\Sigma_0$ with respect to the normal that points into $M$.
\item $\theta^+\ge0$ on $S$ with respect to the normal that points out of $M$.
\item $M$ satisfies the homotopy condition with respect to $\Sigma_0$.
\item $\Sigma_0$ satisfies the cohomology condition.
\end{enumerate}
Then, the following hold:
\begin{enumerate} [(i)]
\item $M \cong [0,\ell]\times\Sigma_0$ for some $\ell>0$.\\
Let $\Sigma_t \cong \{t\} \times \Sigma_0$ with unit normal $\nu_t$ in direction of the foliation.
\item $\chi^+ = 0$ on $\Sigma_t$ for every $t \in [0, \ell]$. 
\item  $\Sigma_t$ is a flat torus with respect to the induced metric for every $t \in [0, \ell]$. 
\item $\mu + J (\nu_t) = 0$ on $\Sigma_t$ for every $t \in [0, \ell]$. In particular, $\mu = |J|$ on $M$.
\end{enumerate}
\end{theorem}


The definitions of the \textsl{cohomology condition} and the \textsl{homotopy condition} are given in Section~\ref{sec:weaklyoutermost}.
The cohomology condition ensures that $\Sigma_0$ does not admit a metric of positive scalar curvature. The homotopy condition holds, for example, if $M$ has \textsl{almost product topology} $M \cong ([0,1] \times \Sigma_0) \, \# \, N$ where $N$ is a closed manifold. It implies that
$\Sigma_0$ is connected. A priori, we allow $S$ to have multiple components.  

The assumptions of Theorem~\ref{thm.global.foliation} are satisfied in the compactified picture (after Lohkamp) described above. Note that Theorem~\ref{thm.global.foliation} provides a relatively simple proof of Theorem~2 in \cite{Lohkamp2016} in dimensions $3 \le n \le 7$. In conjunction with Corollary~2.11 in \cite{Lohkamp2016}, this leads to an alternative proof of Theorem~\ref{thm.STposmass} stated above. Conversely, note that under the assumptions of Theorem~\ref{thm.global.foliation}, the proof of Theorem~2 in \cite{Lohkamp2016} implies that $\mu = |J|$. 

Theorem~\ref{thm.global.foliation} is a global version of the local rigidity result for MOTS obtained in \cite{Galloway} and stated here as Theorem~\ref{thm.local.foliation}. We apply Lemma~\ref{lemma.weakly.outermost} to ensure that the \textsl{weakly outermost} condition (see Section~\ref{sec:preliminaries}) of this local rigidity result holds in our setting.

Imposing a convexity condition on the spacetime second fundamental form $K$, we are able to prove the following, stronger rigidity result. Note also that the boundary conditions are different from those in Theorem~\ref{thm.global.foliation}.

\begin{theorem}\label{thm.rigidity}
Let $(M,g,K)$ be an $n$-dimensional, $3 \le n \le 7$, compact-with-boundary initial data set. Suppose that 
$(M,g,K)$ satisfies the DEC, $\mu \geq |J|$. Suppose also that the boundary can be expressed as a disjoint union $\p M=\Sigma_0\cup S$ of non-empty unions of components such that the following conditions hold:
\begin{enumerate}
\item $\theta^+\le0$ on $\Sigma_0$ with respect to the normal that points into $M$.
\item $\theta^-\ge 2 \, (n-1) \, \epsilon$ on $S$ with respect to the normal that points out of $M$, where $\epsilon=0$ or $\epsilon = 1$.
\item $M$ satisfies the homotopy condition with respect to $\Sigma_0$.
\item $\Sigma_0$ satisfies the cohomology condition.
\item $K+\epsilon \, g$ is $(n-1)$-convex.
\end{enumerate}
Then, the following hold: 
\begin{enumerate}
\item[\rm (i)] $(\Sigma_0,g_0)$ is a flat torus, where $g_0$ is the induced metric on $\Sigma_0$.
\item[\rm (ii)] $(M,g)$ is isometric to $([0,\ell]\times\Sigma_0,dt^2+e^{2 \, \epsilon \, t} \, g_0)$ for some $\ell>0$.
\item[\rm (iii)] $K=(1-\epsilon) \, a \, dt^2-\epsilon \, g$ on $M$, where $a \in C^\infty(M)$ depends only on $t\in[0,\ell]$.
\item[\rm (iv)] $\mu= 0$ and $J=0$ on $M$.
\end{enumerate}
\end{theorem}

The definition of $(n-1)$-convexity is recalled in Section~\ref{sec:preliminaries}. Note that $K + \epsilon \, g$ is $(n-1)$-convex in the special case where it is positive semi-definite. The case $\epsilon = 1$ is relevant in the asymptotically hyperbolic or asymptotically hyperboloidal setting. Theorem~\ref{thm.rigidity} contains the following Riemannian result.

\begin{corollary}\label{cor.riemannian.case}
Let $(M,g)$ be an $n$-dimensional, $3 \le n \le 7$, connected, oriented, compact-with-boundary Riemannian manifold. Suppose that the scalar curvature of $(M,g)$ satisfies $R \ge-n \, (n-1) \, \epsilon$, where $\epsilon=0$ or $\epsilon = 1$. Suppose also that the boundary can be expressed as a disjoint union $\p M=\Sigma_0\cup S$ of non-empty unions of components such that the following conditions hold:
\begin{enumerate}
\item The mean curvature of $\Sigma_0$ in $(M,g)$ with respect to the normal that points into $M$ satisfies $H\le(n-1) \, \epsilon$.
\item The mean curvature of $S$ in $(M,g)$ with respect to the normal that points out of $M$ satisfies $H\ge(n-1) \, \epsilon$.
\item $M$ satisfies the homotopy condition with respect to $\Sigma_0$.
\item $\Sigma_0$ satisfies the cohomology condition.
\end{enumerate} 
Then $(\Sigma_0,g_0)$ is a flat torus, where $g_0$ is the induced metric on $\S_0$. Moreover, $(M,g)$ is isometric to $([0,\ell]\times\Sigma_0,dt^2+e^{2 \, \epsilon \, t} \, g_0)$ for some $\ell>0$.
\end{corollary}

Note the similarity of Corollary~\ref{cor.riemannian.case} with the rigidity result Theorem~1 in \cite{Croke} in the Ricci curvature setting, due to C.~B.~Croke and B.~Kleiner. Theorem~1.1 in \cite{ACG} follows from Corollary~\ref{cor.riemannian.case} in the special case where $\epsilon = 1$. Corollary~\ref{cor.riemannian.case} should also be compared to the results of H.~C.~Jang and the second-named author in \cite{GalJang}, which require an \textsl{outermost} condition. An alternative proof of Corollary~\ref{cor.riemannian.case} may be given using area minimization when $\epsilon = 0$ and minimization of the so-called brane action when $\epsilon = 1$. The MOTS methodology presented in this paper gives a synthetic way of treating both cases, and much more, simultaneously.  \\

We review some background material on MOTS in Section~\ref{sec:preliminaries}. In Section~\ref{sec:weaklyoutermost}, we establish criteria to verify the weakly outermost condition for MOTS. In Section~\ref{sec:mainproof}, we give a proof of Theorem~\ref{thm.global.foliation}. In Section~\ref{sec:further}, we give the proofs of Theorem~\ref{thm.rigidity} and Corollary~\ref{cor.riemannian.case} and also consider some additional results. Finally, in Section~\ref{sec:Minkowski}, we show how to embed the initial data set in Theorem~\ref{thm.rigidity} into a quotient of Minkowski space. \\

{\bf Acknowledgments.} Michael Eichmair is supported by the START-Project Y963-N35 of the Austrian Science Fund.  Gregory J. Galloway acknowledges the support of NSF Grant DMS-1710808. Abraão Mendes is grateful to the University of Miami where much of his work on this project was carried out. He was supported in part by the Coordenação de Aperfeiçoamento de Pessoal de Nível Superior - Brasil (CAPES) - Finance Code 001. The authors would like to thank Dan A.~Lee and Pengzi Miao for useful discussions on the topic of this paper.

\section{Preliminaries} \label{sec:preliminaries}

We recall several results for MOTS that are needed in this paper.

Let $(M, g, K)$ be an initial data set and $\Sigma \subset M$ be a closed MOTS with unit normal $\nu$. 

Let $\{\Sigma_t\}_{|t| < \epsilon}$ be a variation of $\Sigma$, where 
\[
\Sigma_t = \{ \exp_x (t \, \phi (x) \, \nu(x)) : x \in \Sigma\}
\]
for some $\phi \in C^\infty(M)$. We may view the expansion scalar $\theta^+$ of these hypersurfaces as a parameter-dependent function on $\Sigma$. We recall from e.g.~\cite[p.~861]{AnderssonMarsSimon} or \cite[p.~20]{AEM} that 
\begin{align} \label{eq:thetader}
 \frac{d}{dt} \Big|_{t = 0} \theta^+ ( t, \, \cdot \,) = L \, \phi
\end {align}
where 
\[
L \, \phi = - \Delta \phi + 2 \, \langle X, \nabla \phi \rangle + \left( Q - |X|^2 + \text{div} (X) \right)  \phi
\]
and 
\[
Q = \frac{1}{2} \, R_\Sigma  - \frac{1}{2} \, |\chi^+|^2 - \mu - J (\nu).
\]
Here, $\Delta$ is the non-positive definite Laplace-Beltrami operator, $\nabla$ the gradient, $\text{div}$ the divergence, and $R_\Sigma$ the scalar curvature of $(\Sigma, \langle\, \cdot \, , \, \cdot \, \rangle)$. Moreover, $X$ is the tangent field of $\Sigma$ that is dual to the form $K (\nu, \, \cdot \,)$.

If there is a $\phi \in C^\infty(\Sigma)$ with $\phi > 0$ and 
\[
L \, \phi \geq 0,
\]
then $\Sigma$ is called a \textsl{stable} MOTS; cf.~\cite[p.~868]{AnderssonMarsSimon}. We refer in passing to related notions of stability for MOTS with boundary and their applications; see e.g.~\cite[p.~3]{Galloway-Murchadha:2008}, \cite[Section~2]{potpourri}, or \cite[Section~5]{ALY}.

Assume now that $\Sigma$ is a boundary in $M$. More precisely, assume that $\nu$ points towards a top-dimensional submanifold $M_+ \subset M$ such that $\partial M_+ = \Sigma \cup S$ where $S$ is a union of components of $\partial M$. We think of $M_+$ as the region outside of $\Sigma$. Then $\Sigma$ is called an \textsl{outermost MOTS} if there is no closed embedded surface in $M_+$ with $\theta^+ \leq 0$ that is homologous to and different from $\Sigma$. If there is no such surface with $\theta^+ < 0$, then $\Sigma$ is called \textsl{weakly outermost}.

We now state the local rigidity result for MOTS from \cite{Galloway} mentioned in the introduction.

\begin{theorem}[\cite{Galloway}, Theorem~3.1] 
\label{thm.local.foliation}
Let $(M,g,K)$ be an $n$-dimensional, $n \ge 3$, initial data set that satisfies the DEC, $\mu \geq |J |$. Suppose that $\Sigma \subset M$ is a connected weakly outermost MOTS that does not support a metric of positive scalar curvature. There is a neighborhood $U \subset M$ of $\Sigma$ with $U \cap M_+ \cong[0, \delta)\times\Sigma$ such that the following hold:
\begin{enumerate} [(i)]
\item $\Sigma_0 = \Sigma$ where $\Sigma_t \cong \{t\} \times \Sigma$ for every $t \in [0, \delta)$.
\item $\chi^+ = 0$ on $\Sigma_t$ for every $t\in[0,\delta)$ with respect to the unit normal $\nu_t$ in direction of the foliation. In particular, $\Sigma_t$ is a MOTS  for every $t \in [0, \delta)$.
\item The metric induced on $\Sigma_t$ by $g$ is Ricci-flat for every $t\in[0,\delta)$.
\item $\mu +J(\nu_t)=0$ on $\Sigma_t$ for every $t\in[0,\delta)$. In particular, $\mu = |J|$ on $U \cap M_+$.
\end{enumerate}
\end{theorem}
Theorem~\ref{thm.local.foliation} implies the following result related to the topology of apparent horizons. 

\begin{corollary} \label{cor.outermost}
Let $(M,g,K)$ be an $n$-dimensional, $n \ge 3$, initial data set that satisfies the DEC, $\mu \ge|J |$. Suppose that $\Sigma \subset M$ is an outermost MOTS in $(M,g,K)$. Then $\Sigma$ admits a metric of positive scalar curvature.
\end{corollary}

We will apply the following existence result for MOTS.  It was obtained by L.~Andersson and J.~Metzger \cite{AnderssonMetzger2009} in dimension $n = 3$ and then, using different techniques, by the first-named author \cite{Eichmair2009, Eichmair2010} in dimensions $3 \leq n \leq 7$. The approaches in \cite{AnderssonMetzger2009, Eichmair2009, Eichmair2010} are all based on an idea of R.~Schoen to construct MOTS between suitably trapped hypersurfaces by forcing a blow up of the Jang equation. See also \cite{AEM} for a survey of these existence results. 

\begin{theorem} [\cite{AnderssonMetzger2009, Eichmair2009, Eichmair2010}] \label{thm:existenceoutermostMOTS}
Let $(M,g,K)$  be an $n$-dimensional, $3 \le n \le 7$, compact-with-boundary initial data set.  Suppose that the boundary can be expressed as a disjoint union $\p M=\Sigma_{\inn}\cup\Sigma_{\out}$ where $\Sigma_{\inn}, \, \Sigma_{\out}$ are non-empty unions of components of $\partial M$ with $\theta^+\le 0$ on $\Sigma_{\inn}$ with respect to the normal pointing into $M$ and with $\theta^+>0$ on $\Sigma_{\out}$ with respect to the normal pointing out of $M$. Then there is an outermost MOTS in $(M,g,K)$ that is homologous to $\Sigma_{\out}$. 
\end{theorem}

Some of our results require a convexity condition on the spacetime second fundamental form $K$ of the initial data set $ (M, g, K)$. We say that a symmetric $(0, 2)$-tensor field $P$ is \textsl{$(n-1)$-convex} if, at every point, the sum of the smallest $(n-1)$ eigenvalues of $P$ with respect to $g$ is non-negative. In particular, if $P$ is $(n -1)$-convex, then $\text{tr}_\Sigma \, P \geq 0$ for every hypersurface 
$\Sigma \subset M$. This convexity condition has been used by the third-named author in \cite{Mendes} in a related context.

\section{The Weakly Outermost Condition} \label{sec:weaklyoutermost}

In this section, we assume that $(M,g,K)$ is a compact initial data set. While we do not assume a priori that $M$ is topologically a product, we require a more general condition of a similar flavor. Let $\Sigma_0$ be a union of components of $\partial M$. We say that $M$ satisfies the \textsl{homotopy condition} with respect to $\Sigma_0$ provided that  there exists a continuous map $\rho:M\to\Sigma_0$ such that $\rho\circ i:\Sigma_0\to\Sigma_0$ is homotopic to $\id_{\Sigma_0}$ where $i:\Sigma_0 \to M$ is the inclusion map. Since $M$ is connected by assumption, this condition implies that $\Sigma_0$ is connected. Note that this homotopy condition is satisfied if there is a retraction of $M$ onto $\Sigma_0$.

An orientable, closed manifold $N$ of dimension $m$ is said to satisfy the \textsl{cohomology condition} if there are classes $\omega_1,\, \ldots,\, \omega_m\in H^1(N,\Z)$ whose cup product
\begin{align*} 
\omega_1\smile\cdots\smile\omega_m\in H^m(N,\Z)
\end{align*}
is non-zero. Such a manifold $N$ has a component that does not admit a metric of positive scalar curvature; see \cite[Theorem~5.2]{SY2017} and the discussion of \cite[Theorem~2.28]{Lee}. Note that every manifold diffeomorphic to $\T^m$ or, more generally, to $\mathbb{T}^m \, \# \, Q$ with $Q$ oriented and closed satisfies the cohomology condition; cf.~\cite[Theorem~5.1]{SY2017}.

\begin{lemma} \label{lem.cohomology.hypothesis}
Let $M$ be an orientable, compact $n$-dimensional, $n \geq 3$, manifold with boundary. Let $\Sigma_0$ be a union of components of $\partial M$. Suppose that $M$ satisfies the homotopy condition with respect to $\Sigma_0$ and that $\Sigma_0$ satisfies the cohomology condition. Then every closed, embedded hypersurface $\Sigma \subset M$ homologous to $\Sigma_0$ satisfies the cohomology condition.
\end{lemma}

In particular, by the preceding discussion, $\Sigma$ does not support a metric of positive scalar curvature.

\begin{proof}
Let $\rho:M\to\Sigma_0$ be a continuous map such that $\rho\circ i\simeq\id_{\Sigma_0}$, where $i:\Sigma_0 \to M$ is the inclusion map. Let $\omega_1,\, \ldots, \, \omega_{n-1}\in H^1(\Sigma_0,\Z)$ be classes with $\omega_1\smile\cdots\smile\omega_{n-1}\neq0$. Let $j:\Sigma \to M$ be the inclusion map. The map $\sigma=\rho\circ j:\Sigma\to\Sigma_0$ induces a map $H_{n-1}(\Sigma,\Z)\to H_{n-1}(\Sigma_0,\Z)$. Since $\Sigma$ and $\Sigma_0$ are homologous,
\begin{align*}
\sigma_*[\Sigma]=\rho_*(j_*[\Sigma])=\rho_*(i_*[\Sigma_0])=(\id_{\Sigma_0})_*[\Sigma_0]=[\Sigma_0].
\end{align*}
Using this, we can conclude the proof arguing as in \cite[p.~45]{Lee}. Note that
\begin{align*}
\sigma_*([\Sigma]\frown(\sigma^*\omega_1\smile\cdots\smile\sigma^*\omega_{n-1}))=[\Sigma_0]\frown(\omega_1\smile\cdots\smile\omega_{n-1})\neq0.
\end{align*}
In particular, $\sigma^*\omega_1\smile\cdots\smile\sigma^*\omega_{n-1}\neq0$.
\end{proof}

We combine the previous lemma with Corollary~\ref{cor.outermost} and Theorem~\ref{thm:existenceoutermostMOTS} to show that the weakly outermost condition follows from seemingly weaker assumptions.

\begin{lemma}\label{lemma.weakly.outermost}
Let $(M,g,K)$ be an $n$-dimensional, $3 \le n \le 7$,  compact-with-boundary initial data set.  Suppose that $(M,g,K)$ satisfies the DEC, $\mu \ge|J|$. Suppose also that the boundary can be expressed as a disjoint union $\partial M = \Sigma_0 \cup S$ of non-empty unions of components such that the following conditions hold:
\begin{enumerate}
\item $\theta^+\le0$ on $\Sigma_0$ with respect to the normal that points into $M$.
\item $\theta^+\ge0$ on $S$ with respect to the normal that points out of $M$.
\item $M$ satisfies the homotopy condition with respect to $\Sigma_0$.
\item $\Sigma_0$ satisfies the cohomology condition.
\end{enumerate}
Then $\Sigma_0$ is a weakly outermost MOTS in $(M,g,K)$.
\end{lemma}

In the proof of this result below, we compute expansion scalars with respect to different spacetime second fundamental forms. For clarity, we indicate by a subscript which spacetime second fundamental form is used in the computation.

\begin{proof}

First, we show that $\Sigma_0$ is a MOTS. 

Suppose that $\theta_K^+$ is not identically zero on $\Sigma_0$. It follows from \cite[Lemma~5.2]{AnderssonMetzger2009} that there is a hypersurface $\Sigma \subset M$ obtained as a small perturbation of $\Sigma_0$ into $M$ such that $\theta_K^+<0$ on $\Sigma$ with respect to the normal pointing away from $\Sigma_0$. Let $W$ be the connected, compact region bounded by $\Sigma$ and $S$ in $M$. Observe that $\theta_{-K}^+\le0$ on $S$ with respect to the normal that points into $W$ and $\theta_{-K}^+>0$ on $\Sigma$ with respect to the normal that points out of $W$. Note that the initial data set $(W, g, - K)$ satisfies the DEC. Applying Theorem~\ref{thm:existenceoutermostMOTS} to this initial data set, we obtain an outermost MOTS $\widetilde\Sigma$ in $(W,g,-K)$ that is homologous to and disjoint from $\Sigma$. Clearly, $\widetilde\Sigma$ is homologous to $\Sigma_0$. By 
Lemma~\ref{lem.cohomology.hypothesis}, $\widetilde\Sigma$ does not support a metric of positive scalar curvature. This contradicts Corollary~\ref{cor.outermost} applied to the initial data set $(W,g,-K)$. Thus $\Sigma_0$ is a MOTS in $(M, g, K)$.

Next, we show that $\Sigma_0$ is a weakly outermost MOTS. 

Suppose, by contradiction, that $\Sigma_0$ is not weakly outermost. Then there is a hypersurface $\Sigma \subset M$ homologous to $\Sigma_0$ such that $\theta_K^+<0$ on $\Sigma$ with respect to the normal that points away from $\Sigma_0$. We may assume that each component of $\Sigma$ is homologically non-trivial in $M$. Let $W$ be the compact region in $M$ bounded by $\Sigma$ and $S$. Assume first that $W$ is connected. Applying Theorem~\ref{thm:existenceoutermostMOTS}, we obtain an outermost MOTS $\widetilde\Sigma$ in $(W,g,-K)$ homologous to $\Sigma$. As before, by 
Lemma~\ref{lem.cohomology.hypothesis}, $\widetilde\Sigma$ does not support a metric of positive scalar curvature. This contradicts Corollary~\ref{cor.outermost}. In the general case, we apply this argument separately to each component of $W$. At least one of the outermost MOTS obtained in this way does not carry a metric of positive scalar curvature. Again, this contradicts Corollary~\ref{cor.outermost}.
\end{proof}

\section{Proof of Theorem~\ref{thm.global.foliation}} \label{sec:mainproof}

Theorem~\ref{thm.global.foliation} may be viewed as a global version of Theorem~\ref{thm.local.foliation}. We emphasize that Theorem~\ref{thm.global.foliation} does not require the \textsl{weakly outermost} assumption. 

We start with the following observation. 

\begin{lemma} \label{lemma.4.1} Under the assumptions of Theorem~\ref{thm.global.foliation}, there is a neighborhood of $\Sigma_0$ in $M$ diffeomorphic to $[0,\delta) \times \Sigma_0$ such that the leaves $\Sigma_t \cong \{t\} \times \Sigma_0$ satisfy properties (ii)-(iv) of the conclusion of Theorem~\ref{thm.global.foliation}.
\begin {proof}

According to Lemma~\ref{lemma.weakly.outermost}, $\Sigma_0$ is a weakly outermost MOTS in $(M, g, K)$. We may apply Theorem~\ref{thm.local.foliation}. We claim that the foliation from  Theorem~\ref{thm.local.foliation} has the asserted properties. It only remains to show that each $\Sigma_t$ is isometric to a flat torus. To see this, note the estimate $b_1(\S_t) \geq n$ for the first Betti number of $\S_t$. This follows from the cohomology condition, Poincaré duality, and the fact that $H^n(M,\bbZ)$ is torsion free. Conversely, by a classical result of Bochner, see e.g.~\cite[p.~208]{Petersen}, it holds that $b_1(\S_t) \le n$ with equality if and only if $\S_t$ is isometric to a flat torus. 
\end {proof}

\end{lemma}

\begin{proof}[Proof of Theorem~\ref{thm.global.foliation}] We use $\nu$ to denote the unit normal field of the foliation $\{\Sigma_t\}_{t \in [0, \delta)}$ from Lemma~\ref{lemma.4.1}. Note that the divergence of $\nu$ evaluated on $\Sigma_t$ is equal to the mean curvature of $\Sigma_t$. Since every leaf $\Sigma_t$ is a MOTS, we see that the divergence of $\nu$ is bounded. By the divergence theorem, 
\[
\text{vol} (\Sigma_t) = \text{vol} (\Sigma_0) + \int_{U_t} \text{div} (\nu)
\]
where $U_t \cong [0, t] \times \Sigma_0$ is the collar between $\Sigma_0$ and $\Sigma_t$. This argument shows that $\vol (\Sigma_t)$ is bounded independently of $t \in [0, \delta)$. 

Note that the second fundamental form of each $\Sigma_t$ is bounded independently as well, since the null second fundamental form of each $\Sigma_t$ vanishes. 

To proceed, we briefly recall a standard fact. For convenience of exposition, we extend $(M, g)$ across its boundary to a homogeneously regular manifold. Given $C > 0$, there is a small constant $r > 0$ with the following property. Let $\Sigma \subset M$ be a closed and two-sided surface whose second fundamental form is bounded by $C$. Let $p \in M$ be such that $\Sigma \cap B_r(p) \neq \emptyset$. Then $\Sigma$ contains the graph of a function $f : \{ y \in \R^{n-1} : |\, y \, |\ < r\} \to \R$ with $| \, f (0) \, |\ < r$, $| \, Df \, |\ \leq 1$, and $|\, D^2 f \, |\ \leq 2 \, C$, where an appropriately rotated geodesic coordinate system with center at $p$ is used to identify $B_{2 \, r} (p)$ with the Euclidean ball $\{x \in \R^n : |\, x \, |\ < 2 \, r\}$. In fact, it is possible to choose the geodesic coordinate system so that $\Sigma \cap B_r(p)$ is covered by such graphs.

It follows from these facts that the leaves $\{\Sigma_t\}_{t \in [0, \delta)}$ have a smooth \textsl{immersed} limit $\Sigma_\delta$ as $t \nearrow \delta$. We use an idea of L.~Andersson and J.~Metzger \cite{AnderssonMetzger2009} to show that $\Sigma_\delta$ is embedded. For if not, we can find for every $\eta > 0$ a leaf $\Sigma_t$ and $p \in M$ such that $\Sigma_t \cap B_r(p)$ contains the graphs of two functions $f_1, \, f_2$ with the properties stated above and such that $| \, f_1 (0) - f_2 (0) \, |\ < \eta$. In fact, we can arrange for the layer between these graphs to lie to the outside of $\Sigma_t$. Arguing exactly as in Section~6 of \cite{AnderssonMetzger2009}, if $\eta > 0$ is sufficiently small, it is possible to glue in a neck to connect $\Sigma_t$ across this layer to obtain a surface with non-positive expansion and negative expansion around the neck. By flowing this surface outward at the speed of its expansion as in Lemma~5.2 of \cite{AnderssonMetzger2009}, one obtains a surface homologous to $\Sigma_t$ with everywhere negative expansion. This contradicts the fact that $\Sigma_0$ is weakly outermost. 

It is easy to see from the proof of Theorem~\ref{thm.local.foliation} that $\{\Sigma_t\}_{t \in [0, \delta]}$ is a foliation. Recall that $M$ and $\Sigma_\delta \cong \Sigma_0$ are connected. By the strong maximum principle as in e.g.~\cite[Proposition~3.1]{AshtekarGalloway} or \cite[Proposition~2.4]{AnderssonMetzger2009}, we have that $\Sigma_\delta = S$ if $\Sigma_\delta \cap S \neq \emptyset$. Note that the assumptions of the theorem continue to hold if we replace $\Sigma_0$ by $\Sigma_\delta$ and $M$ by the complement of $U$ in $M$. The result now follows by a continuity argument.
\end{proof}

\begin {example} 
The following example shows that there is still a fair amount of flexibility in the initial data sets covered by Theorem~\ref{thm.global.foliation}.  

Let $\mathbb{R}^3_1$ be Minkowski space with standard coordinates $t, x, y, z$. Consider the box  $\mathcal{B} = \{ (x,y,z) : 0 \le x \le 1,  0 \le y \le 1, 0 \le z \le 1\}$ in the $t = 0$ slice. Let $f : \mathcal{B} \to \R$ be a smooth function that vanishes near the boundary of $\mathcal{B}$ and whose graph is spacelike in $\mathbb{R}^3_1$. We identify opposite sides in the $x$ and the $y$ coordinate to obtain an initial data set $(M, g, K)$ with $M \cong \mathbb{T}^2 \times [0, 1]$. Let $\Sigma_0$ be the torus corresponding to $t = z = 0$. Note that $(M,g,K)$ satisfies the conditions of Theorem~\ref{thm.global.foliation}. The foliation in the conclusion of Theorem~\ref{thm.global.foliation} arises from intersecting the graph of $f$ with the null hypersurfaces $\mathcal{H}_c : t = z + c$. This can be understood using the following standard argument; see e.g.~\cite[Appendix A]{ChruArea}. The hypersurfaces $\mathcal{H}_c$  are totally geodesic null hypersurfaces, i.e.~each has vanishing null second fundamental form with respect to any null vector field $K_c$ tangent to
$\mathcal{H}_c$.  Since $K_c$ is orthogonal to every spacelike cross section, it follows that all these cross sections have vanishing null second fundamental form. In particular, they are MOTS.  
Moreover, again because $\mathcal{H}_c$ is totally geodesic, the induced metric on every spacelike cross section is invariant under the flow generated by $K_c$. It follows that any two such cross sections are isometric.
\end {example} 

\section{Proof of Theorem~\ref{thm.rigidity} and Further Consequences} \label{sec:further}

As in the proof of Lemma~\ref{lemma.weakly.outermost}, we will compute several quantities with respect to different spacetime second fundamental forms. We indicate by subscript the second fundamental form that is used. 

Theorem~\ref{thm.rigidity} follows from the local rigidity result below.

\begin{lemma} \label{lem.local.rigidity}
Assumptions as in Theorem~\ref{thm.rigidity}. Then $(\Sigma_0,g_0)$ is a flat torus, where $g_0$ is the metric on $\Sigma_0$ induced by $g$. Moreover, there is a neighborhood $U$ of $\Sigma_0$ in $M$ such that the following hold:
\begin{enumerate} [(i)]
\item $(U,g)$ is isometric to $([0,\delta)\times\Sigma_0,dt^2+e^{2 \, \epsilon \, t} \, g_0)$, for some $\delta>0$.
\item $K=(1-\epsilon) \, a \, dt^2-\epsilon \, g$ on $U$, where $a$ depends only on $t\in[0,\delta)$.
\item $\mu=0$ and $J=0$ on $U$.
\end{enumerate} 
\end{lemma}

\begin{proof}

By assumption, 
\[
\theta_K^-=H-\tr_{S} \, K\ge2 \, (n-1) \, \epsilon
\]
on $S$, where $H$ is the mean curvature of $S$ with respect to the normal pointing out of $M$. Using also the assumption that $K+\epsilon \, g$ is $(n-1)$-convex, we obtain
\[
H\ge\tr_S \, K+2 \, (n-1) \, \epsilon = \tr_S \, K + (n-1) \, \epsilon + (n- 1) \, \epsilon \ge (n-1)\, \epsilon.
\]
Therefore,
\[
\theta_K^+=H+\tr_SK \geq (n - 1) \, \epsilon + \tr_S K \ge0
\]
on $S$. It follows from Lemma~\ref{lemma.weakly.outermost} that $\Sigma_0$ is a weakly outermost MOTS in $(M,g,K)$. By Theorem~\ref{thm.local.foliation}, $\Sigma_0$ has a neighborhood $U_1 \cong [0,\delta_1)\times\Sigma_0$ in $M$ such that the following hold: 
\begin{enumerate} [--]
\item We have that 
\[
g = \phi_1^2 \, ds^2+g_1(s)
\]
on $U_1$, where $g_1(s)$ is the metric on $\Sigma_1(s)\cong\{s\}\times\Sigma_0$ induced by $g$.
\item Every leaf $\Sigma_1(s)$ is a MOTS. In fact, 
\[
0 = \chi_K^+(s)= A_1(s) + K|_{\Sigma_1(s)},
\]
where $A_1(s)$ is the second fundamental form of $\Sigma_1(s)$ in $M$ computed with respect to the unit normal $\nu_1(s)$ in direction of the foliation.
\item We have that 
\[
\mu_K+J_K(\nu_1(s))=0.
\]
\end{enumerate}
Consider now the initial data set $(M,g,P)$, where 
\[
P=-K-2 \, \epsilon \, g.
\]
Note that $(M,g,P)$ satisfies the DEC. In fact, 
\[
\mu_P-|J_P|=\mu_K-|J_K|+2 \, (n-1) \, (\tr \, K+n \, \epsilon) \, \epsilon \ge 0
\]
where we have used the assumption that $K+\epsilon \, g$ is $(n-1)$-convex.

By assumption, $\theta_K^+=H+\tr_{\Sigma_0}K\le0$ on $\Sigma_0$, where $H$ is the mean curvature of $\Sigma_0$ with respect to the normal that points into $M$. It follows that
\[
H\le-\tr_{\Sigma_0}K\le(n-1) \, \epsilon
\]
on $\Sigma_0$ and thus 
\[
\theta_P^+=H+\tr_{\Sigma_0}P=H-\tr_{\Sigma_0}K-2 \, (n-1) \, \epsilon\le0.
\]
Also, 
\[
\theta^+_P=\theta_K^{-}-2 \, \epsilon \, (n-1)\ge0
\]
on $S$. By Lemma~\ref{lemma.weakly.outermost}, $\Sigma_0$ is a weakly outermost MOTS in $(M,g,P)$. It follows from Theorem~\ref{thm.local.foliation} that there is a neighborhood $U_2$ of $\Sigma_0$ in $M$ diffeomorphic to $[0,\delta_2)\times\Sigma_0$ for some $\delta_2>0$, such that the following hold: 
\begin{enumerate} [--]
\item We have that 
\[
g = \phi_2^2 \, dt^2+g_2(t)
\]
on $U_2$, where $g_2(t)$ is the metric on $\Sigma_2(t)\cong\{t\}\times\Sigma_0$ induced by $g$. 
\item Every leaf $\Sigma_2(t)$ is a MOTS. In fact, 
\[
0 = \chi_P^+(t)= A_2(t) + P|_{\Sigma_2(t)},
\]
where $A_2(t)$ is the second fundamental form of $\Sigma_2(t)$ in $M$ computed with respect to the unit normal $\nu_2(t)$ in direction of the foliation.
\item $(\Sigma_2(t),g_2(t))$ is Ricci flat.
\item We have that 
\[
\mu_P+J_P(\nu_2(t))=0.
\]
\end{enumerate}
Decreasing $\delta_2>0$, if necessary, we may assume that $U_2\subset U_1$. Fix $t \in (0, \delta_2)$ and note that $\Sigma_1(s)\cap\Sigma_2(t)\neq\emptyset$ for some $s\in(0,\delta_1)$, since $\Sigma_2(t)\subset U_2\subset U_1$ and $\Sigma_1(0)\cap\Sigma_2(t)=\Sigma_0\cap\Sigma_2(t)=\emptyset$. Let
\[
s_0=s_0(t)=\inf\{s\in(0,\delta_1) :\Sigma_1(s)\cap\Sigma_2(t)\neq\emptyset\}
\]
and note that $\Sigma_1(s_0)\cap\Sigma_2(t)\neq\emptyset$. In particular, $s_0>0$. Also, $\Sigma_1(s)\cap\Sigma_2(t)=\emptyset$ for all $s\in[0,s_0)$. This means that $\Sigma_2(t)$ is contained in the region outside of $\Sigma_1(s_0)$.

The mean curvature of $\Sigma_1(s_0)$ is given by 
\[
H_1(s_0)=\tr\chi_K^+(s_0)-\tr_{\Sigma_1(s_0)}K=-\tr_{\Sigma_1(s_0)}K\le (n-1) \, \epsilon.
\]
For the mean curvature of $\Sigma_2(t)$, we have the estimate
\[
H_2(t)=\tr\chi_P^+(t)-\tr_{\Sigma_2(t)}P=\tr_{\Sigma_2(t)}K+2 \, (n-1) \, \epsilon\ge (n-1) \, \epsilon.
\]
In particular, 
\[
H_1(s_0)\le H_2(t)
\]
so that 
\[
\Sigma_1(s_0)=\Sigma_2(t)
\]
by the maximum principle. 

We see that the foliations $\{\Sigma_1 (s)\}_{s \in [0, \delta_1)}$ and $\{\Sigma_2 (t)\}_{t \in [0, \delta_2)}$ are the same after reparametrization. Below, we will denote this foliation of a neighborhood $U$ of $\Sigma_0$ in $M$ by $\{\Sigma (t)\}_{t \in [0, \delta)}$. Note that $\chi_K^+=0$ and $\chi_P^+=0$ on each leaf $\Sigma (t)$. Let $\nu (t)$ be the unit normal of $\Sigma (t)$ in direction of the foliation, $g(t)$ the induced metric, $A(t)$ the second fundamental form with respect to $\nu (t)$, and $\phi$ the lapse function of the foliation. By \eqref{eq:thetader}, we have that 
\[
\ds 0=\frac{\p\theta_K^+}{\p t}=-\Delta\phi+2\langle X_K,\nabla\phi\rangle+ \left(\divv X_K-|X_K|^2 \right) \, \phi
\]
where
\[
\ds Q_K =\frac{1}{2} \, R_{\Sigma (t)}-\left(\mu_K+J_K(\nu (t))\right)-\frac{1}{2} \, |\chi_K^+|^2
\]
vanishes. Arranging terms as in \cite[(2.9)]{GS}, we obtain that
\[
\divv(X_K-\nabla\ln\phi)-|X_K-\nabla\ln\phi|^2=0.
\]
Integrating both sides of this equation over $\Sigma (t)$ and applying the divergence theorem, we obtain that 
\[
X_K=\nabla\ln\phi
\]
on $\Sigma (t)$. By the same argument, we find 
\[
X_P=\nabla\ln\phi.
\] 

From the definition of $P$,
\[
X_P=-X_K \qquad \qquad  \mu_P=\mu_K+2 \, (n-1) \, (\tr \, K+n \, \epsilon) \, \epsilon \qquad \qquad J_P=-J_K.
\]
Thus, 
\[
\mu_K=-J_K(\nu (t))=J_P(\nu (t))=-\mu_P=-\mu_K-2 \, (n-1) \, (\tr \, K+n \, \epsilon) \, \epsilon
\]
and 
\[
\nabla\ln\phi=X_K=-X_P=-\nabla\ln\phi.
\] 
It follows that 
\[
|J_K| \le\mu_K = -(n-1) \, (\tr \, K+n \, \epsilon) \, \epsilon \le 0 \qquad \text{ and } \qquad \nabla\ln\phi=0.
\] 
From this, we conclude that 
\[
\mu_K = 0 \qquad \qquad \qquad J_K = 0 \qquad \qquad \qquad (\tr \, K+n \, \epsilon) \, \epsilon=0
\]
on $U$. Moreover, the lapse function $\phi$ is constant on $\Sigma (t)$ for every $t\in[0,\delta)$. 

Using that 
\[
0 = \chi_K^+=A(t) + K|_{\Sigma (t)} \qquad \text{ and } \qquad 0 = \chi_{P}^+=A (t)-K|_{\Sigma(t)}-2 \, \epsilon \, g (t),
\]
we obtain 
\[
A (t)=\epsilon \, g (t)=-K|_{\Sigma (t)}
\]
and thus 
\[
g=dt^2+e^{2 \, \epsilon  \, t} \, g_0.
\]
Using also that $K(\nu (t), \, \cdot \, )|_{\Sigma (t)}=0$ since $X_K=0$, we see that
\[
K=a \, dt^2-\epsilon \, g (t)
\]
on $U$. If $\epsilon=1$, we have $\tr \, K=-n$. Thus $a = -1$ and hence $K=-g$ on $U$. If $\epsilon=0$, we use that $d(\tr \,  K)=\divv K$ (since $J_K=0$)  to see that $a$ is constant on every leaf $\Sigma (t)$. Finally, the same argument as in the proof of Lemma~\ref{lemma.4.1} shows that $(\Sigma_0,g_0)$ is a flat torus. This completes the proof.
\end{proof}

\begin{proof} [Proof of Theorem~\ref{thm.rigidity}] 
Let 
\[
\ell = \sup \{ \delta : \text{the conclusion of Lemma~\ref{lem.local.rigidity} holds with this value of $\delta >0$}\}.
\]
Note that $\ell < \infty$ since $M$ is compact. Reasoning the embeddedness of the final sheet as in the proof of Theorem~\ref{thm.global.foliation}, we see that $(M, g)$ is isometric to the warped product $([0,\ell] \times\Sigma_0,dt^2+e^{2 \, \epsilon \, t} \, g_0)$. Moreover, we see that (iii) and (iv) of Theorem~\ref{thm.rigidity} hold.
\end{proof}

\begin{proof}[Proof of Corollary~\ref{cor.riemannian.case}]
Let $K=-\epsilon \, g$ and note that $(M,g,K)$ satisfies the DEC. In fact, a straightforward calculation gives that
\[
\ds\mu_K=\frac{1}{2} \, \left( R+n \, (n-1) \, \epsilon \right)\ge0 \quad \mbox{ and } \quad J_K=0.
\]
The expansion $\theta_K^+$ of $\Sigma_0$ in $(M,g,K)$ computed with respect to the normal that points into $M$ satisfies 
\[
\theta_K^+=H+\tr_{\Sigma_0}K=H- (n-1) \, \epsilon\le0.
\]
The expansion $\theta_K^-$ of $S$ in $(M,g,K)$ computed with respect to the normal that points out of  $M$ satisfies
\[
\theta_K^-=H-\tr_{S} K=H+(n-1) \, \epsilon\ge 2 \, (n-1) \, \epsilon.
\]
Thus Theorem~\ref{thm.rigidity} applies to $(M, g, K)$ and gives the assertion.
\end{proof}

The next theorem establishes a rigidity result under the boundary conditions of Theorem~\ref{thm.global.foliation}, assuming a volume minimizing condition on $\Sigma_0$.

\begin{theorem}\label{thm.volume.minimizing}
Let $(M,g,K)$ be an $n$-dimensional, $3\le n\le 7$,  compact-with-boundary initial data set. Suppose that $(M,g,K)$ satisfies the DEC, $\mu\ge|J|$. Suppose also that the boundary can be expressed as a disjoint union $\p M=\Sigma_0\cup S$ of non-empty unions of components such that the following conditions hold:
\begin{enumerate}
\item $\theta^+\le0$ on $\Sigma_0$ with respect to the normal that points into $M$.
\item $\theta^+\ge0$ on $S$ with respect to the normal that points out of $M$.
\item $M$ satisfies the homotopy condition with respect to $\Sigma_0$.
\item $\Sigma_0$ satisfies the cohomology condition.
\item $K$ is $(n-1)$-convex.
\item $\Sigma_0$ is volume minimizing in $(M,g)$.
\end{enumerate}
Then, the following hold:
\begin{enumerate}
\item[\rm (i)] $(\Sigma_0,g_0)$ is a flat torus, where $g_0$ is the metric on $\Sigma_0$ induced by $g$.
\item[\rm (ii)] $(M,g)$ is isometric to $([0,\ell]\times\Sigma_0,dt^2+g_0)$, for some $\ell>0$.
\item[\rm (iii)] $K=a \, dt^2$ on $M$, where $a$ depends only on $t\in[0,\ell]$.
\item[\rm (iv)] $\mu=0$ and $J=0$ on $M$.
\end{enumerate}
\end{theorem}

As shown in the following example, Theorem~\ref{thm.volume.minimizing} fails to hold if one drops either the 
$(n-1)$-convexity assumption or the volume minimizing assumption.

\begin {example}
Let $(\Sigma_0,g_0)$ be the square flat $(n-1)$-torus. Let $(M,g)$ be the cylinder $([0,\ell]\times\Sigma_0,dt^2+e^{2 \, \epsilon \, t} \, g_0)$ and $K=-\epsilon \, g$, where $\epsilon=-1$ or $\epsilon=1$. The second fundamental form of $\Sigma_t=\{t\}\times\Sigma_0$ in $(M,g)$ with respect to the  normal in direction of increasing values of $t$ is given by $A(t)=\epsilon \, e^{2\, \epsilon \,  t} \, g_0$. Then $(M,g,K)$ satisfies all the assumptions of Theorem~\ref{thm.volume.minimizing} except for the volume minimizing assumption in the case where $\epsilon=-1$ and the $(n-1)$-convexity assumption in the case where  $\epsilon=1$.
\end {example}

\begin{proof}[Proof of Theorem~\ref{thm.volume.minimizing}]
It follows from Lemma~\ref{lemma.weakly.outermost} that $\Sigma_0$ is weakly outermost. Then, by Theorem~\ref{thm.local.foliation}, there exists a neighborhood $U$ of $\Sigma_0$ in $M$ diffeomorphic to $[0,\delta)\times\Sigma_0$ for some $\delta>0$, such that:
\begin{enumerate} [--]
\item We have that 
\[
g=\phi^2 \, dt^2+g(t)
\]
on $U$, where $g(t)$ is the metric on $\Sigma(t)\cong\{t\}\times\Sigma_0$ induced by $g$.
\item Every leaf $\Sigma(t)$ is a MOTS. In fact, 
\[
\chi_K^+(t)=A(t) + K|_{\Sigma(t)}=0,
\]
where $A(t)$ is the second fundamental form of $\Sigma(t)$ in $(M,g)$.
\item $(\Sigma(t),g(t))$ is Ricci flat.
\item We have that
\[
\mu_K+J_K(\nu(t))=0,
\]
where $\nu(t)$ is the unit normal field on $\Sigma(t)$ in direction of increasing values of $t$.
\end{enumerate}
Since $\tr_{\Sigma(t)}K\ge0$, we have 
\[
H(t)\le H(t)+\tr_{\Sigma(t)}K=\tr\chi_K^+(t)=0
\]
where $H(t)$ is the mean curvature of $\Sigma(t)$ in $(M,g)$. The first variation formula for the volume of $(\Sigma(t),g(t))$ gives
\[
\ds\frac{d}{dt}\Vol(\Sigma(t),g(t))=\int_{\Sigma(t)}\phi \, H(t) \, d\vol_{g(t)}\le0.
\] 
In particular,
\begin{align} \label{eq.first.variation.volume}
\Vol(\Sigma(t),g(t))\le\Vol(\Sigma_0,g_0)
\end{align}
for every $t \in [0, \delta)$. Since $\Sigma_0$ is volume minimizing by assumption, we obtain
\[
\Vol(\Sigma(t),g(t))=\Vol(\Sigma_0,g_0)
\]
for all $t\in[0,\delta)$. Then, by \eqref{eq.first.variation.volume}, we have $H(t)=0$, which implies $\tr_{\Sigma(t)}K=0$, for each $t\in[0,\delta)$. Therefore $\theta_K^+=\theta_K^-=0$ on $\Sigma(t)$, for each $t\in[0,\delta)$.

As in the proof of Lemma~\ref{lem.local.rigidity}, the first variation of $\theta_K^+$ gives that $X_K=\nabla\ln\phi$ 
on $\Sigma(t)$. On the other hand, the first variation of $\theta_K^-=\theta_{-K}^+$ gives that $X_{-K}=\nabla\ln\phi$ 
on $\Sigma_t$. Proceeding as in the proof of Lemma~\ref{lem.local.rigidity}, we obtain the following \textsl{local rigidity}:
\begin{enumerate} [--]
\item $(\Sigma_0,g_0)$ is a flat torus, where $g_0$ is the metric on $\Sigma_0$ induced by $g$.
\item $(U,g)$ is isometric to $([0,\delta)\times\Sigma_0, dt^2+g_0)$.
\item $K=a\, dt^2$ on $U$, where $a$ depends only on $t\in[0,\delta)$.
\item $\mu_K=0$ and $J_K=0$ on $U$.
\end{enumerate} 
Observe that $\Sigma(t)$ is also volume minimizing in $(M,g)$. The assertion follows from this local rigidity as in the proof of Theorem~\ref{thm.rigidity}.
\end{proof}

\section{Embedding of the initial data into a quotient of Minkowski space} \label{sec:Minkowski}

In this section we show how, under the assumptions of Theorem~\ref{thm.rigidity}, $(M,g)$ can be isometrically embedded into a quotient of the Minkowski spacetime in a such way that $K$ is exactly its second fundamental form. This, together with Theorem~\ref{thm.rigidity}, characterizes the geometry -- both intrinsic and extrinsic -- of the initial data set $(M,g,K)$ under natural conditions. The same holds under the assumptions of Theorem~\ref{thm.volume.minimizing}.

\begin{theorem}\label{thm.embedding}
Assumptions as in Theorem~\ref{thm.rigidity} or Theorem~\ref{thm.volume.minimizing}. 
There is an isometric embedding of $(M,g)$ into a quotient of Minkowski space in a such way that $K$ is its second fundamental form.
\end{theorem}

Consider the Minkowski spacetime $\R_1^n$ of dimension $n+1$, i.e.~$\R\times\R\times\R^{n-1}$ with the Lorentzian metric 
\[
g_M=-dt^2+dr^2+d\x^2
\]
where $d\x^2$ is the standard Euclidean metric on $\R^{n-1}$. 

Given a smooth function $t:\R\to\R$, define $r:\R\to\R$ by
\[
\ds r(s)=\int_0^s\sqrt{1+t'(\sigma)^2} \, d\sigma.
\]
Consider the spacelike hypersurface
\[
N_0=\{(t(s), \, r(s), \, x) : s\in\R, \, x\in\R^{n-1}\}\subset\R_1^n.
\]
Note that $(N_0,h_0)$ is isometric to $(\R\times\R^{n-1},ds^2+d\x^2)$ where $h_0$ is the metric induced by $g_M$. Straightforward calculations show that the second fundamental form of $N$ in $\R_1^n$ with respect to $\frac{\p}{\p s}$ is given by 
$P=b \, ds^2$, where $b:\R\to\R$ is the function 
\[
\ds b=\frac{t''}{\sqrt{1+t'{}^2}}.
\]

Now, consider the hyperbolic space $\Hy^n$ of dimension $n$, that is, the $n$-manifold $\R_+\times\R^{n-1}$ endowed with the metric 
\[
\ds g_H=\frac{1}{x_0^2}\, (dx_0^2+d\x^2).
\] 
Using the change of variables $x_0$ to $s=- \ln x_0$, we may write 
\[
g_H=ds^2+e^{2 \, s} \, d\x^2.
\]
Thus $(N_1,h_1)=(\R\times\R^{n-1},ds^2+e^{2 \, s} \, d\x^2)$ is isometric to hyperbolic space.
 
Consider the hypersurface 
\[
H=\{(t, \, r, \, x)\in\R\times\R\times\R^{n-1}:-t^2+r^2+|x|^2=-1, \, t<0\},
\]
where $|x|^2=x_1^2+\cdots+x_{n-1}^2$. Recall that $H$ with the metric induced by $g_M$ is isometric to $\Hy^n$ and that the second fundamental form of $H$ in $\R_1^n$ with respect to the future directed unit normal is given by $-g_M|_{H}$. 

\begin{proof}[Proof of Theorem~\ref{thm.embedding}]
From Theorem~\ref{thm.rigidity} and Theorem~\ref{thm.volume.minimizing}, we know that $(M,g)$ is isometric to $([0,\ell]\times\Sigma_0,ds^2+e^{2 \, \epsilon\, t} \, g_0)$, where $(\Sigma_0,g_0)$ is a flat torus and $K=(1-\epsilon) \, a \, dt^2-\epsilon g$ for some function $a:[0,\ell]\to\R$. 
Therefore, $(M,g)$ is isometric to a quotient of $([0,\ell]\times\R^{n-1},ds^2+e^{2 \, \epsilon \, t} \, d\x^2)$. In the case where $\epsilon=0$, we can take $t:[0,\ell]\to\R$ to be the solution of
\[
\ds\frac{t''}{\sqrt{1+t'{}^2}}=a
\]
with initial condition $t(0) = 0$ and $t' (0) = 0$. 
Therefore, identifying $(N_0,h_0)$ with $(\R\times\R^{n-1},ds^2+d\x^2)$, it follows from the above remarks that we can embed $(M,g)$ into a quotient of $\R_1^n$ in a such way that the second fundamental form of $M$ is given by $P=a \, ds^2=K$. In the case where $\epsilon=1$, it suffices to identify $(N_1,h_1)$ with $(H,h)$, where $h$ is the metric on $H$ induced by $g_M$.
\end{proof}

\bibliographystyle{amsplain}
\bibliography{bibliography}

\end{document}